\documentclass[aps,prl,twocolumn,showpacs,groupedaddress,notitlepage,nobibnotes,superscriptaddress,floatfix,longbibliography,preprintnumbers]{revtex4-2}
\usepackage{bbm}
\usepackage{mathrsfs}
\usepackage{epsfig}
\usepackage{soul}
\usepackage{graphicx}
\usepackage{amsfonts}
\usepackage{amsthm}
\usepackage[figuresright]{rotating}
\usepackage{amssymb}
\usepackage{amsmath}
\usepackage{dcolumn}
\usepackage{physics}
\usepackage{bm}
\usepackage{verbatim}
\usepackage[normalem]{ulem}
\usepackage[ruled,vlined,linesnumbered]{algorithm2e}
\usepackage{setspace}
\usepackage[dvipsnames]{xcolor}
\usepackage[colorlinks,linkcolor=blue,anchorcolor=blue,citecolor=blue,urlcolor=blue]{hyperref}
\usepackage{stackengine}
\usepackage{cleveref,enumitem}
\AddToHook{cmd/appendix/before}{\crefalias{section}{appendix}}
\usepackage{cleveref-forward}

\newtheorem{theorem}{Theorem}
\newtheorem{lemma}[theorem]{Lemma}

\newtheorem{fact}[theorem]{Fact}

\newtheorem{definition}{Definition}[section]

\newcommand{\PRLsec}[1]{\textit{#1}---}

\setlist[itemize]{leftmargin=*}
\setlist[enumerate]{leftmargin=*}

\newcommand{\lr}[1]{\left(#1\right)}
\newcommand{\E}[1]{\mathbb{E}\left[#1\right]}
\newcommand{\Esub}[2]{\mathbb E_{#1}\left[#2\right]}
\newcommand{\eps}{\epsilon}
\newcommand{\EE}{\mathbb{E}}
\newcommand{\Var}{\mathrm{Var}}

\newcommand{\cA}{\mathcal{A}}
\newcommand{\cB}{\mathcal{B}}
\newcommand{\cC}{\mathcal{C}}

\newcommand{\cE}{\mathcal{E}}

\newcommand{\cG}{\mathcal{G}}

\newcommand{\cP}{\mathcal{P}}
\newcommand{\cR}{\mathcal{R}}

\newcommand{\cT}{\mathcal{T}}
\newcommand{\cU}{\mathcal{U}}

\newcommand{\argmin}{\mathop{\mathrm{argmin}}}

\newcommand{\KL}{\mathrm{KL}}
\newcommand{\bern}{\operatorname{Bern}}

\def\ba#1\ea{\begin{align}#1\end{align}}

\begin{document}
\title{Optimal Lower Bounds for Hamiltonian Simulation}
\author{Alexander Zlokapa}
\thanks{These authors contributed equally to this work.}
\author{Richard R. Allen}
\thanks{These authors contributed equally to this work.}
\author{Aram W. Harrow}
\affiliation{Center for Theoretical Physics -- a Leinweber Institute, MIT}
\preprint{MIT-CTP/6075}
\begin{abstract}
For Hamiltonian $H = \sum_j h_j$, we prove asymptotically tight lower bounds on the gate and query complexities of simulating time evolution on a quantum computer. Our bounds hold for arbitrary term norms $\|h_j\|$, time $t$, and trace-distance error $\epsilon$. The matching upper bound (known as composite qDRIFT) consists of high-order Trotterization of the large terms and a randomized first-order Trotterization of the small terms. Unlike prior work that chooses worst-case $\|h_j\|$ to encode the computation of parity or other Boolean functions in time evolution, our proof is elementary and based on a local, bounded-degree classical Hamiltonian. Our work suggests that for many physical systems (e.g., power-law interactions), gate count must scale polynomially in $1/\epsilon$, contrary to the complexity suggested by counting coherent oracle queries such as those in the block-encoding model.
\end{abstract}

\maketitle

A central task in quantum simulation is to identify the best time evolution algorithms for many-body quantum systems under various realistic assumptions, such as the locality of interactions~\cite{haah2021quantum,childs2019nearly,childs2021theory}, symmetries~\cite{tran2021faster}, or properties of the initial state~\cite{csahinouglu2021hamiltonian,Zhao_2022,CB24}. Many natural problems have Hamiltonians of the form $H = \sum_{j=1}^L h_j$, where $L$ is large, $h_j$ is $k$-local (on spins, fermions and/or bosons), and many terms have small $\norm{h_j}$ relative to the largest terms; examples include power-law interactions and second-quantized electronic structure Hamiltonians~\cite{tran2019locality,ouyang2020compilation,motta2021low}.
Three options are commonly considered for simulating such Hamiltonians.
\begin{enumerate}
    \item \emph{Truncate small terms.} Instead of evaluating $e^{-iHt}$, the Hamiltonian is truncated to $\widetilde H$ containing only the largest Hamiltonian terms. A standard simulation algorithm (e.g., quantum signal processing or Trotterization~\cite{lloyd1996universal,childs2019nearly,childs2021theory,childs2012hamiltonian,berry2015simulating,low2017optimal,low2019hamiltonian}) is used to evaluate the unitary evolution $e^{-i\widetilde H t}$.
    \item \emph{Randomized simulation.} Randomized methods like the qDRIFT algorithm~\cite{campbell2019random,childs2019faster,chen2021concentration} implement Trotter steps $e^{-ih_j\Delta t}$ for Hamiltonian terms $h_j$ sampled proportional to $\norm{h_j}$. This has worse $t$ dependence than deterministic methods, but it avoids the error introduced by truncation.
    \item \emph{Hybrid algorithm.} One can optimize the tradeoff between the above two options: deterministically evolve the largest terms and randomly sample the smaller terms each Trotter step. This gives better results than truncation or randomization alone~\cite{ouyang2020compilation,hagan2023composite,jin2025partially,rajput2022hybridized,kiss2023importance,pocrnic2024composite}. 
\end{enumerate}
Our main result shows that the final option is in fact asymptotically \emph{optimal} in both the gate and query models for a natural and simple family of Hamiltonians. More concretely, we consider the family of local, bounded-degree Hamiltonians of the form
\begin{align}\label{eq:h-main-indep}
    H = \sum_{j=1}^L a_j h_j, \; a_1 \geq a_2 \geq \cdots \geq a_L > 0, \;\, \sum_{j=1}^L a_j = 1,
\end{align}
where each $h_j$ is a distinct, $k$-local Hermitian operator with $\norm{h_j} = 1$ (e.g., a Pauli string) and $k \geq 2$. The normalization of the coefficient vector $\vec a$ is without loss of generality by rescaling $t$. In fact, we will let $H$ be time-dependent with a coefficient vector $\vec a$. Specifically, we take the same conditions on $\vec a$ as \cref{eq:h-main-indep}, but set
\begin{align}\label{eq:h-main}
	H(\tau) = \sum_{j=1}^L a_j h_j(\tau).
\end{align}

We show lower bounds on the cost of implementing a channel $\cE_H$ that approximates ideal time evolution $\cT \mathrm{exp}(-i\int_0^t H(\tau)d\tau)$. Notably, our lower bounds will be optimal for \emph{any} choice of $a_1,\dots,a_L$. This departs from prior lower bounds for Hamiltonian simulation, which choose worst-case coefficients~\cite{berry2007efficient,childs2009limitations,berry2014exponential,berry2015hamiltonian,atia2017fast}. Before introducing our results and comparing to existing no-fast-forwarding theorems, we motivate our work by recalling well-known quantum algorithms for simulating time evolution.

\PRLsec{Prior algorithms.}%
Algorithms for Hamiltonian simulation have improved steadily since Feynman's initial proposal~\cite{feynman1982simulating}. Two of the main algorithmic paradigms are Trotter-Suzuki~\cite{lloyd1996universal,childs2019nearly,childs2021theory} and block encoding (known variously as LCU, QSP, or QSVT)~\cite{childs2012hamiltonian,berry2015simulating,low2017optimal,low2019hamiltonian}. These have nearly optimal cost with respect to simulation time $t$ and error $\eps$ but not with respect to the number of terms $L$ in the Hamiltonian.

Trotterization (a.k.a.~Trotter-Suzuki) approximates the unitary $e^{-iHt}$ with a unitary $U$ via a product formula, which can take advantage of locality via commutator structure~\cite{lloyd1996universal,childs2019nearly,childs2021theory}. At first order, an $R$-step Trotterization of the Hamiltonian \cref{eq:h-main-indep} gives
\begin{align}
    U = \lr{\exp\Big(-i\frac{t}{R} a_L h_L\Big) \cdots \exp\Big(-i\frac{t}{R} a_1 h_1\Big)}^R.
\end{align}
A total of $O(L t^2 / \eps)$ two-qubit gates suffices to achieve $\eps$ operator-norm error (the convention $\sum_{j} a_j = 1$ means we have absorbed the overall norm of $H$ into $t$). Here, we assume that $t\geq 1$; otherwise replace $t$ with $\max(t,1)$ in the costs for Trotter and LCU below.
More generally, higher-order Trotterization approximates time evolution using product formulas of the form
\begin{align}
    U = \prod_{i=1}^\Upsilon \prod_{j=1}^L e^{-ib_{(i,j)} a_{\pi_i(j)} h_{\pi_i(j)}t},
\end{align}
for real coefficients $b_{(i,j)}$ and a permutation $\pi_i$ that orders operator summands within stage $i$. If taken to sufficiently high order, Trotterization simulates the Hamiltonian of \cref{eq:h-main-indep} using
\begin{align}\label{eq:trotter}
    O\lr{\frac{Lt^{1+o(1)}}{\epsilon^{o(1)}}}
\end{align}
gates. Beyond Trotterization, another commonly considered algorithm, quantum signal processing, approximates the unitary $e^{-iHt}$ with a unitary $U$ using
\begin{align}\label{eq:qsp}
    O\lr{L \lr{t + \frac{\log1/\epsilon}{\log\log 1/\epsilon}}}
\end{align}
gates~\cite{childs2012hamiltonian,berry2015simulating,low2017optimal,low2019hamiltonian}.
The factor of $L$ above is not widely discussed; it arises from the cost of the block encoding. One might hope that this overhead reflects a poorly optimized block encoding strategy which could be generically improved. Our results show that this is not possible, unless we assume the Hamiltonian has additional structure.

For Hamiltonians with large $L$, such as the second-quantized electronic-structure Hamiltonian, neither of these methods is ideal. For an alternative approach, observe that the algorithms above prepare a \emph{unitary} approximation to $e^{-i H t}$. If, instead, we consider a \emph{randomized} algorithm that prepares a channel approximation $\cE(\rho) = \sum_k p_k U_k \rho U_k^\dagger$ to the corresponding unitary channel $\cU(\rho) = e^{-iHt}\rho e^{iHt}$, we can ensure that the cost of any individual $U_k$ in the channel is independent of $L$.
Indeed, this is the key insight of qDRIFT~\cite{campbell2019random}, which is perhaps the most natural algorithm for simulating Hamiltonians with many small terms.
The algorithm prepares a mixture of product formulas, $\{U_k\}$, each of which has the form of a first-order Trotter formula. Specifically, a length-$R$ product formula $U_k = V_R \cdots V_1$ is prepared by evolving for timestep $t/R$ under a random $h_j$ sampled with probability $a_j$, i.e.,
\begin{align}
  V_r = e^{-i(t/R)h_{j_r}}, \
  j_r \sim \vec{a}
\end{align}
Since $\E{h} = H$ and thus $\E{V_r} \approx e^{-iHt/R}$ in the large-$R$ limit, repeating this $R$ times can be shown to approximate evolution under $e^{-iHt}$ via a mixing lemma~\cite{campbell2017shorter,hastings2016turning}. Due to its structure, the gate count of this algorithm is similar to first-order Trotterization. Crucially, however, the cost is now \emph{independent} of $L$: $O(t^2/\eps)$ two-qubit gates suffice to achieve $\epsilon$ diamond-norm error.
While the $L$ dependence has been removed, the dependence on $t$ and $\epsilon$ is no longer as favorable as \cref{eq:trotter} or \cref{eq:qsp}.

A substantial body of follow-up work has explored how far randomization can be pushed beyond the original qDRIFT algorithm, including analysis of random permutations of product formulas~\cite{childs2019faster} and individual random product formulas~\cite{chen2021concentration}.
An important improvement to qDRIFT lies in a hybrid approach~\cite{ouyang2020compilation,hagan2023composite,jin2025partially,rajput2022hybridized,kiss2023importance,pocrnic2024composite}. In its simplest form, the hybrid approach chooses some $0 \leq K \leq L$ and deterministically evolves the largest $K$ terms via Trotterization, then randomly samples the remaining terms via qDRIFT. By construction, this hybrid algorithm cannot perform worse than any of the above algorithms. For example, combining Trotter and qDRIFT (following \cite{hagan2023composite}, see also Appendix \ref{sec:ub-HW}) gives a cost of
\begin{align}\label{eq:composite}
    O\lr{\min_{0 \leq K \leq L} \lr{Kt  \left( \frac{t}{\epsilon} \right)^{o(1)}+ \frac{t^2 \lambda_K^2}{\epsilon}}}
\end{align}
for tail mass
\begin{align}\label{eq:lambda}
    \lambda_K = \sum_{j > K} a_j \leq 1.
\end{align}
We refer to this approach as \emph{composite qDRIFT}. The same asymptotic upper bound also holds for the piecewise-constant, time-dependent Hamiltonian of \cref{eq:h-main}, provided the coefficients $\vec a$ are fixed and the Hamiltonian changes only on intervals of length $\Theta(1)$, by simulating each interval separately and concatenating the resulting channels.
These hybrid approaches are based on the plausible idea that, since Trotter scales well with $\sum_j a_j$ and qDRIFT scales well with $L$, we should divide up the Hamiltonian terms into a Trotter (or LCU) piece and a qDRIFT piece. 

One might hope for improvements beyond the performance of composite qDRIFT in~\cref{eq:composite}. For example, one might attempt to design more sophisticated, ``higher-order'' versions of qDRIFT which retain the $L$ independence but achieve better scaling with $t$ and $\eps$, as in deterministic methods.  However, the natural ways to attempt this~\cite{FSKKE22,NBA24,WW25-Trotter,W25-qFLO} lead not to sampling algorithms but to quasiprobability methods which (a) yield only estimates of expectation values rather than the time-evolved state itself; and
(b) still require a runtime $\sim \min(Lt, t^2)$. (We omit the $\eps$ dependence here, since the error metrics are different for state evolution and observable estimation.)

\PRLsec{Summary of results.}%
In this work, we show that composite qDRIFT is essentially tight in all parameters, ruling out the existence of any higher-order qDRIFT: to approximate the ideal simulation $\cU$ to $\epsilon$ trace distance error, we establish gate and query lower bounds that match the upper bound of \cref{eq:composite}.
The hard instances underlying these lower bounds are especially simple: at every time the Hamiltonian is local and classical. Indeed, for the query lower bound, the Hamiltonian is also non-interacting; to give a two-qubit gate lower bound, we introduce interactions on pairs of qubits, but the pairs are otherwise non-interacting.

\begin{theorem}[Gate lower bound]\label{thm:gate}
Given $t$, $\vec{a}$, and sufficiently small $\eps$, there exists a time-dependent Hamiltonian $H(\tau)=\sum_j a_jh_j(\tau)$ with $\|h_j(\tau)\|=1$ such that the following holds. Letting $\cU_H$ denote the quantum channel implementing ideal time evolution $\cT\mathrm{exp}\big(-i\int_0^t H(\tau) \dd\tau\big)$, if a channel $\cE_H$ using $G$ two-qubit gates satisfies
\begin{align}
    \frac{1}{2} \| \cE_H(\rho) - \cU_H(\rho) \|_1 \le \eps
\end{align}
for all density matrices $\rho$, then 
\begin{align}\label{eq:lb}
    G \geq \Omega\lr{\min_{0 \leq K \leq L} \lr{Kt + \frac{t^2\lambda_K^2}{\epsilon}}}.
\end{align}
\end{theorem}
\noindent
\Cref{thm:gate} implies, in particular, that any simulation of the form
\begin{align}\label{eq:mix}
    \cE_H(\rho) = \sum_r p_r V_r \rho V_r^\dagger
\end{align}
for unitaries $V_r$ and probabilities $p_r$ must have at least one $V_r$ with at least $G$ two-qubit gates. Algorithms that satisfy \cref{eq:mix} include Trotterization, QSP, qDRIFT, and composite qDRIFT. We refer to \Cref{sec:gate} for a more formal definition of the gate count of a channel.

Before describing the proof of \Cref{thm:gate}, we give a similar result in the query model. The query model is commonly used to bound complexity by counting the number of uses of an oracle. A familiar example is in quantum signal processing (QSP), where one uses the select and prepare oracles for unitary $h_j$ (e.g., Pauli strings) defined as
\begin{align}
    \mathrm{SEL} = \sum_{j=1}^L \ketbra{j} \otimes h_j, \quad \mathrm{PREP} : \ket{0} \to \sum_{j=1}^L \sqrt{a_j}\ket{j}
\end{align}
to prepare a block-encoding of $H$~\cite{low2019hamiltonian}. Indeed, optimal algorithms for quantum simulation are known in terms of the query complexity to these oracles~\cite{low2017optimal,low2019hamiltonian}.

When given a Hamiltonian as a list of Pauli strings and coefficients, however, the select and prepare oracles do \emph{not} capture the full complexity of quantum simulation. For example, given $H = \sum_j a_j h_j$, it is possible to block-encode $H$ with only $O(1)$ queries to SEL and PREP, assuming $\sum_j a_j = 1$. However, this ignores the cost of constructing the oracles: the SEL unitary costs up to $O(L)$ gates to implement. This is why we report the cost of QSP in \cref{eq:qsp} as scaling with $L$.

In this work, we use an access model that significantly weakens the power of the oracles to more accurately reflect the cost of simulation.
We use a \emph{classical} oracle $O_H$ that takes as input an index $j \in [L]$ and time $\tau$ and returns a classical description of the coefficient $a_j$ and local operator $h_j(\tau)$ for the $j$-th term in the Hamiltonian.  In particular, it cannot be queried in superposition. We prove a query lower bound which matches the gate lower bound (without assuming any particular representation of $\cE_H$):

\begin{theorem}[Query lower bound]\label{thm:query}
Let $t, \vec a, \eps, \cU_H$ be as defined in \Cref{thm:gate}.  There is a distribution over Hamiltonians $H(\tau)$, with corresponding classical oracles $O_H$, such that any algorithm $\cE_H$ making $Q$ classical queries to $O_H$ and satisfying
\begin{align}
    \Esub{H}{\frac{1}{2} \| \cE_H(\rho) - \cU_H(\rho) \|_1} \le \epsilon
\end{align}
for all density matrices $\rho$ must satisfy
\begin{align}
    Q \ge \Omega\lr{\min_{0 \leq K \leq L} \lr{Kt + \frac{t^2\lambda_K^2}{\epsilon}}}.
\end{align}
\end{theorem}

This classical oracle model is most natural when considering Hamiltonians with many distinct terms that are computed classically because computing them in superposition would add significantly to the cost.  This is the situation, for example, in the second-quantized electronic-structure problem. Here the Hamiltonian is
\ba
 H = \sum_{i,j,k,l \in [N]} h_{ijkl} c_i^\dag c_j^\dag c_k c_l + \sum_{i,j \in [N]} h_{ij} c_i^\dag c_j.
\label{eq:electrons}
\ea
Each $h_{ijkl}$ and $h_{ij}$ coefficient is computed classically, and this list of $L=O(N^4)$ coefficients is used to determine the quantum circuit.  Our query lower bounds apply in this setting.  By contrast, they would not apply to first-quantized simulations in which the integrals determining $h_{ijkl}$ are computed in superposition.

\PRLsec{Proof outline}%
We now give a sketch of the underlying argument used for the proofs of \Cref{thm:gate,thm:query}. We start with a time-independent Hamiltonian and then prove lower bounds at asymptotically late times by introducing piecewise time-dependence (similarly to~\cite{haah2021quantum}). For this sketch, it suffices to consider the Hamiltonian
\begin{align}\label{eq:hsimp}
    H = \sum_{j=1}^L a_j \ketbra{1}_j, \quad \sum_{j=1}^L a_j = 1.
\end{align}
In the formal proofs provided in \Cref{sec:gate,sec:query}, we make only minor modifications to \cref{eq:hsimp}: we use $\ketbra{11}_{2j-1,2j}$ to obtain a two-qubit gate lower bound, and we insert random signs $a_j \to \pm a_j$ to show a query lower bound. Up to these technical details, however, the argument remains largely the same.

Start by assuming $a_1 = \cdots = a_L=1/L$ for simplicity. We will argue for a lower bound of
\begin{align}\label{eq:lbqdrift}
    \Omega(\min\{Lt, t^2/\epsilon\}).
\end{align}
The first challenge is to choose an initial state $\ket\psi$ on which time evolution is hard.  
The Mandelstam-Tamm bound says that an initial state $\ket{\psi}$ with energy standard deviation
\begin{align}
    \Delta_\psi(H) = \sqrt{\bra{\psi} H^2 \ket{\psi} - (\bra{\psi} H \ket{\psi})^2}
\end{align}
requires time $\Omega(1/\Delta_\psi(H))$ to evolve to an orthogonal state. To see error quickly, we should thus choose a state with large energy variance. For the Hamiltonian in \cref{eq:hsimp}, this variance is maximized by a cat state
\begin{align}\label{eq:cat}
    \ket{\rm cat} = \frac{| 0^L \rangle + | 1^L \rangle}{\sqrt 2}, \quad \Delta_{\rm cat}(H) = \frac{1}{2}
\end{align}
since time evolution adds phases coherently. So far, our argument follows \cite{chen2021concentration}. We deviate from prior work by considering evolution of $\ket{\psi} = \otimes_{m=1}^M \ket{\psi_m}$ for cat states $\ket{\psi_m}$ acting on disjoint qubits. We divide the Hamiltonian into corresponding patches $H=\sum_{m=1}^M H_m$, such that the sum of $a_j$ coefficients in each patch is $\sim 1/M$; this ensures $\Delta_{\psi_m}(H_m) \sim 1/M$. Suppose the algorithm approximates evolution for some small time $\tau$. If it does not apply any gates to the patch, then the fidelity against the true time-evolved state is
\begin{align}\label{eq:guess}
    \abs{\bra{\psi_m}e^{-i\tau H_m}\ket{\psi_m}}^2 \approx 1 - \tau^2 \Delta_{\psi_m}^2(H_m) \approx 1 - \frac{\tau^2}{M^2}
\end{align}
for $\tau \ll 1$. The algorithm incurs a similar error in the query model, since time evolution introduces either a positive or negative relative phase in the cat state depending on the sign of the interaction; if the algorithm did not query $H_m$ to determine the sign, one can bound the quality of its best guess similarly to \cref{eq:guess}. Since the patches do not interact in either the Hamiltonian or the initial state $\ket{\psi} = \otimes_{m=1}^M \ket{\psi_m}$, the fidelity of the final state is $(1-\tau^2/M^2)^M \approx 1-\tau^2/M$. Since the ideal time-evolved state $\ket{\psi(\tau)}$ is pure, the trace distance between the algorithm's output $\rho$ and $\psi(\tau) = \ketbra{\psi(\tau)}$ is bounded by the Fuchs-van de Graaf inequality as
\begin{align}
\frac{1}{2}\norm{\psi(\tau) - \rho}_1 \geq    1 - \lr{1 - \frac{\tau^2}{M}} = \frac{\tau^2}{M}  .
\end{align}
Hence, an algorithm that achieves $\epsilon$ error in trace distance must apply gates to at least
\begin{align}
    M \geq \frac{\tau^2}{\epsilon}
\end{align}
patches, or make a similar number of queries. When the coefficients $a_j$ are not uniform, one can use a similar argument on the smallest $L-K$ terms to obtain a lower bound of $\tau^2\lambda_K^2/\epsilon$ instead of $\tau^2/\epsilon$. We then choose $\tau$ appropriately and make the Hamiltonian piecewise time-dependent to address asymptotically large times, obtaining a lower bound of $\Omega(t^2\lambda_K^2/\epsilon)$ for simulating the tail of the Hamiltonian. To obtain the final lower bound of \cref{eq:lb}, it remains to analyze the largest terms. These are comparatively straightforward: if $ta_j$ is not small, then the cat state acquires a detectable relative phase upon evolution.

\PRLsec{In what sense is composite qDRIFT ``optimal"?}%
The history of Hamiltonian simulation algorithms is littered with claims of optimality or near-optimality, followed by further improvements~\cite{berry2007efficient,berry2014exponential,berry2015hamiltonian,low2017optimal,low2019nearly}.  This is because ``optimality'' means optimal scaling in the worst case with respect to some choice of parameters and access model, and different settings call for different models. In our case, we study optimality with respect to $t$, $\epsilon$, and the Hamiltonian term magnitudes $a_1,\dots,a_L$. These parameters determine the cost of most simulation algorithms, including qDRIFT, Trotterization, and LCU/QSP.  In other words, these algorithms (qDRIFT, etc.) do not exploit Hamiltonian structure beyond knowledge of $t$, $\epsilon$, $a_1,\ldots,a_L$, although in some cases their analyses do.

Our lower bounds are more fine-grained than previous no-fast-forwarding arguments~\cite{berry2007efficient,childs2009limitations,berry2014exponential,berry2015hamiltonian,atia2017fast,haah2021quantum}, which prove optimality in $t$ and $\epsilon$ but use worst-case $\vec a$. 
These prior lower bounds chose $\vec a$ so that Hamiltonian time evolution encodes a fine-tuned computational task, such as evaluating parity or another Boolean function~\cite{berry2007efficient,childs2009limitations,berry2014exponential,berry2015hamiltonian,atia2017fast,haah2021quantum}.  As a result, they did not exclude possibilities such as a hypothetical ``higher-order qDRIFT'' that might achieve subquadratic scaling in $t$ without a dependence on $L$.

Another conclusion from our lower bound is that improving beyond the performance of composite qDRIFT will require exploiting structure beyond the coefficient vector $\vec a$.   This could take the form of assumptions about the input state~\cite{csahinouglu2021hamiltonian,Zhao_2022,CB24,king2026quantum}, the spatial locality of the Hamiltonian~\cite{haah2021quantum,childs2019nearly,childs2021theory}, or even specializing to a specific Hamiltonian such as the electronic structure problem~\cite{Low_2025}.  Our bounds do not directly apply to these algorithms, but are still relevant in the cases when these more specialized algorithms use general-purpose Hamiltonian simulation as a subroutine.  For example, tensor factorization methods for the electronic structure problem combine a problem-specific fermionic swap network with a generic LCU/QSP outer loop~\cite{lee2021even,Low_2025}.  In short, the value of our lower bound, in a world where one works with specific and not worst-case Hamiltonians, is that we show that beating \cref{eq:lb} requires using additional assumptions about the Hamiltonian or input state beyond the coefficient vector $\vec a$.

\paragraph{Acknowledgments.}
AWH and RRA were supported by the US DOE, Office of Science, National Quantum Information Science Research Centers, Co-design Center for Quantum Advantage (C2QA) under contract number DE-SC0012704.  AWH was also supported by the MIT-IBM Watson AI lab and NSF grant  PHY-2325080.
AZ is funded by a Hertz fellowship and the Simons Foundation 
(MP-SIP-00001553, AWH).  We thank Minh Tran for many helpful discussions and Milad Marvian for sharing his group's draft of \cite{Marvian26}.
  
\bibliography{References-clean}
\clearpage
\onecolumngrid
\appendix
\setcounter{secnumdepth}{3}

\section{Upper bound from Hagan--Wiebe~\cite{hagan2023composite}}\label{sec:ub-HW}

Here we justify the claim that \cref{eq:composite} is an achievable cost for composite qDRIFT.  This fact is essentially proved in \cite{hagan2023composite} but the notation there requires some translation.

To avoid a conflict with the locality parameter $k$ used elsewhere in this paper, we denote the product-formula order parameter by $p$; thus the deterministic part uses an order $2p$ Suzuki formula.  Let
\begin{align}
    \Upsilon = 2\cdot 5^{p-1}
\end{align}
be the number of stages in the order $2p$ Suzuki formula.

\begin{fact}[Hagan--Wiebe, Theorem 2.1]\label{fact:HW}
Let $H=A+B$ be a time-independent Hamiltonian, where
\begin{align}
    A=\sum_{\ell=1}^{L_A} \alpha_\ell A_\ell,
    \qquad
    B=\sum_{\ell=1}^{L_B} \beta_\ell B_\ell,
    \qquad
    \norm{A_\ell}=\norm{B_\ell}=1,
    \qquad
    \alpha_\ell,\beta_\ell\geq 0,
\end{align}
and define $\lambda_B=\sum_{\ell=1}^{L_B}\beta_\ell$.  For an integer $N_B\geq 1$, the order $2p$ composite channel of Hagan--Wiebe, obtained by combining an order $2p$ Suzuki formula on the $A$ part with qDRIFT on the $B$ part, approximates the exact evolution channel for $H$ for time $t$ to diamond-norm error at most $\epsilon$ using at most
\begin{align}\label{eq:HW-cost}
    C_{\mathrm{HW}}(A,B,t,\epsilon,2p) \leq \Upsilon\left(\Upsilon L_A+N_B\right)
    \left\lceil \frac{(\Upsilon t)^{1+1/(2p)}4^{1/(2p)}}{\epsilon^{1/(2p)}} \lr{\frac{\Upsilon \alpha_{\mathrm{comm}}(A,2p) + \alpha_{\mathrm{comm}}(\{A,B\},2p)}{2p+1}}^{1/(2p)} + \frac{4\Upsilon\lambda_B^2t^2}{N_B\epsilon}\right\rceil
\end{align}
operator exponentials. Here
\begin{align}
    \alpha_{\mathrm{comm}}(F,2p) = \sum_{\gamma_1,\ldots,\gamma_{2p+1}}
    \lr{\prod_{r=1}^{2p+1} f_{\gamma_r}}
    \norm{[F_{\gamma_{2p+1}},[F_{\gamma_{2p}},\ldots,[F_{\gamma_2},F_{\gamma_1}]\ldots]]},
\end{align}
for $F=\sum_\gamma f_\gamma F_\gamma$ with $\norm{F_\gamma}=1$, and $\alpha_{\mathrm{comm}}(\{A,B\},2p)$ denotes the same commutator sum restricted to nested commutators containing at least one term from $A$ and at least one term from $B$. Note that $\{A,B\}$ in this expression is not an anticommutator; also, the sign convention in \cite{hagan2023composite} is opposite to ours, but the bound is unchanged.  
\end{fact}

We now apply \Cref{fact:HW} to our Hamiltonian family of the form
\begin{align}
    H=\sum_{j=1}^L a_jh_j,
    \qquad
    a_1\geq a_2\geq \cdots \geq a_L>0,
    \qquad
    \sum_{j=1}^L a_j=1.
\end{align}
We translate notation as
\begin{align}
    L_A=K\in\{0,1,\ldots,L\},
    \qquad
    A = \sum_{j=1}^K a_jh_j,
    \qquad
    B = \sum_{j>K} a_jh_j,
    \qquad
    \lambda_B=\lambda_K=\sum_{j>K}a_j.
\end{align}
To bound the commutator quantities appearing in \cref{eq:HW-cost}, we note that each $h_j$ has
operator norm one, and thus any $(2p+1)$-fold nested commutator has norm at most $2^{2p}$; this gives
\begin{align}
    \alpha_{\mathrm{comm}}(A,2p) \leq 2^{2p}\left(\sum_{j\leq K}a_j\right)^{2p+1} = 2^{2p}(1-\lambda_K)^{2p+1}.
\end{align}
Similarly, the mixed commutator sum is bounded by the total weight of all ordered
$(2p+1)$-tuples containing at least one index from $A$ and at least one index from $B$:
\begin{align}
    \alpha_{\mathrm{comm}}(\{A,B\},2p)
    &\leq
    2^{2p}
    \left[
        1-(1-\lambda_K)^{2p+1}-\lambda_K^{2p+1}
    \right].
\end{align}
We can thus bound one of the terms in \cref{eq:HW-cost} as
\begin{align}\label{eq:HW-comm-bound}
    \Upsilon\alpha_{\mathrm{comm}}(A,2p)
    +\alpha_{\mathrm{comm}}(\{A,B\},2p)
    &\leq
    2^{2p}
    \left[
        \Upsilon(1-\lambda_K)^{2p+1}
        +1-(1-\lambda_K)^{2p+1}-\lambda_K^{2p+1}
    \right] \leq
    2^{2p}\Upsilon,
\end{align}
where the final inequality uses $0\leq \lambda_K\leq 1$ and $\Upsilon\geq 1$.

Since the cost at the endpoints $K=0$ and $K=L$ are given individually by qDRIFT and Trotter, we only need to show the upper bound for $1\leq K<L$. Choosing $N_B=K$ in \Cref{fact:HW}, we find by \cref{eq:HW-comm-bound} that
\begin{align}
    C_{\mathrm{HW}}(A,B,t,\epsilon,2p)
    &\leq
    \Upsilon(\Upsilon+1)K
    \left\lceil
        4^{1/(2p)}
    \Upsilon^{1+1/(2p)}
    \left(\frac{2^{2p}\Upsilon}{2p+1}\right)^{1/(2p)} t\left(\frac{t}{\epsilon}\right)^{1/(2p)}
        +
        \frac{4\Upsilon\lambda_K^2t^2}{K\epsilon}
    \right\rceil
\end{align}
and thus
\begin{align}\label{eq:fixed-p-bound-middle}
    C_{\mathrm{HW}}(A,B,t,\epsilon,2p)
    =
    O_p\left(
        Kt\left(\frac{t}{\epsilon}\right)^{1/(2p)}
        +
        \frac{t^2\lambda_K^2}{\epsilon}
    \right).
\end{align}
Optimizing over $K$ gives
\begin{align}\label{eq:composite-fixed-p}
    C_{\mathrm{comp}}
    =
    O_p\left(
        \min_{0\leq K\leq L}
        \left\{
            Kt\left(\frac{t}{\epsilon}\right)^{1/(2p)}
            +
            \frac{t^2\lambda_K^2}{\epsilon}
        \right\}
    \right).
\end{align}
Since $p$ is arbitrary, this equivalently means that for every fixed $\eta>0$ one may choose a
fixed product-formula order $p$ large enough so that $1/(2p)\leq \eta$, yielding our claimed upper bound of
\begin{align}
    C_{\mathrm{comp}}
    =
    O\left(
        \min_{0\leq K\leq L}
        \left\{
            Kt\left(\frac{t}{\epsilon}\right)^{o(1)}
            +
            \frac{t^2\lambda_K^2}{\epsilon}
        \right\}
    \right).
\end{align}

\section{Lower bound preliminaries}\label{appendix: prelims}

In this section, we define the hard Hamiltonians used to prove~\Cref{thm:gate,thm:query}. In both cases, the Hamiltonians will be local and classical. The hard Hamiltonian for the query lower bound will be completely non-interacting, while the hard Hamiltonian for the gate lower bound will be interacting only on individual pairs of qubits (so as to prove a two-qubit gate lower bound, allowing arbitrarily many single-qubit gates).

Before defining the hard time-dependent Hamiltonians, we introduce hard time-independent Hamiltonians of the form in~\cref{eq:h-main-indep}, which can be used to prove the desired lower bounds for $t = \Theta(1)$. Then, we will describe a simple, piecewise-constant generalization of these to prove~\Cref{thm:gate,thm:query} also for asymptotically late times $t$. 

\begin{definition}[Time-independent Hamiltonian for gate lower bound]\label{def: time_ind_gate_rra} Let $\pi/2 \ge a_1 \ge a_2 \ge \cdots \ge a_L > 0$. Consider a $2L$ qubit system, with qubits grouped into disjoint pairs $(A_1, B_1), \dots, (A_L, B_L)$. For $j \in [L]$, let $\ketbra{11}_j$ denote the projector onto the $\ket{11}$ state of the pair $(A_j, B_j)$. Define 
\begin{align}
    H_{2q}(\vec a) = \sum_{j=1}^L a_j \ketbra{11}_j.
\end{align}
\end{definition}

\begin{definition}[Time-independent Hamiltonian for query lower bound]\label{def: time_ind_query_rra}
    Let $\pi/2 \ge a_1 \ge a_2 \ge \cdots \ge a_L > 0$ and let $s \in \{\pm 1\}^L$. Consider an $L$ qubit system. For $j \in [L]$, let $\ketbra{1}_j$ denote the projector onto the $\ket{1}$ state of the $j$-th qubit. Define
    \begin{align}
        H_{1q}(\vec a, s) = \sum_{j=1}^L s_j a_j \ketbra{1}_j.
    \end{align}
\end{definition}
\noindent Note that we could equivalently define $H_{1q}(\vec a, s) = \sum_{j=1}^L s_j \tilde{a}_j Z_j$ for $\tilde{a}_j = -a_j / 2$ and $Z_j$ the Pauli-$Z$ operator on the $j$-th qubit. It is necessary to consider a \emph{family} of Hamiltonians for the query lower bound, rather than a fixed Hamiltonian, because otherwise an algorithm could ``hard-code'' the coefficient vector $\vec a$ and simulate the Hamiltonian without any oracle queries. Here, our family is especially simple: we only vary the signs of each Hamiltonian term.

We now define time-dependent generalizations of these two Hamiltonians. The generalization is simple: we keep the coefficient vector $\vec a$ fixed but change the Hamiltonian terms piecewise in time, so that after each unit time step the Hamiltonian acts on a fresh block of qubits. More precisely, we have the following:

\begin{definition}[Time-dependent Hamiltonian for gate lower bound\label{def: time_dep_gate_rra}] Let $t$ be a positive integer, and let $\pi/2 \ge a_1 \ge a_2 \ge \cdots \ge a_L > 0$. Consider a $2Lt$ qubit system, with qubits grouped into $t$ blocks, and within each block $r \in [t]$ grouped into disjoint pairs $(A_1^{(r)}, B_1^{(r)}), \dots, (A_L^{(r)}, B_L^{(r)})$. Let $H_{2q}^{(r)}(\vec a)$ denote the Hamiltonian in~\Cref{def: time_ind_gate_rra} for the $r$-th block of qubits. For $\tau \in (0,t]$, define the time-dependent Hamiltonian
\begin{align}
    H_{2q}(\vec a, \tau)=H_{2q}^{(\lceil \tau \rceil)}(\vec a).
\end{align}
\end{definition}

\begin{definition}[Time-dependent Hamiltonian for query lower bound\label{def: time_dep_query_rra}] Let $t$ be a positive integer, let $\pi/2 \ge a_1 \ge a_2 \ge \cdots \ge a_L > 0$, and let $s = (s^{(1)}, \dots, s^{(t)})$ with $s^{(r)} \in \{\pm 1\}^L$ for each $r \in [t]$. Consider an $Lt$ qubit system, with qubits grouped into $t$ blocks. Let $H_{1q}^{(r)}(\vec a,s^{(r)})$ denote the Hamiltonian in~\Cref{def: time_ind_query_rra} for the $r$-th block of qubits. For $\tau \in (0,t]$, define the time-dependent Hamiltonian
\begin{align}
    H_{1q}(\vec a, s, \tau) = H_{1q}^{(\lceil \tau \rceil)}(\vec a,s^{(\lceil \tau \rceil)}) .
\end{align}
\end{definition}

\noindent We can now formally introduce the query model we consider.

\begin{definition}[Classical oracle access]\label{def: oracle}
	A classical query algorithm for a Hamiltonian family $H(\tau)=\sum_{j=1}^L a_j h_j(\tau)$ has access to a deterministic classical function $O_H$. A query consists of an ordinary classical input $(j,\tau)\in [L]\times(0,t]$, and the oracle returns a classical description of $a_j$ and $h_j(\tau)$.
\end{definition}
\noindent The algorithm may be randomized and adaptive, with later queries depending on previous answers, but it never queries $O_H$ in superposition. After at most $Q$ oracle calls, the algorithm outputs a quantum channel $\cE_H$; conditioned on a transcript $r$, the output channel may depend only on $r$ and on the algorithm's internal randomness. For the time-independent query hard instances $H_{1q}(\vec a,s)$, the coefficients $a_j$ and projectors $\ketbra{1}_j$ are fixed and known, so $O_H$ is equivalent to the sign oracle $O_s(j)=s_j$. For the time-dependent query hard instances $H_{1q}(\vec a,s,\tau)$, this becomes the block-indexed sign oracle $O_s(b,j)=s_j^{(b)}$, where $b=\lceil \tau \rceil$. We count each call to $O_s$ as one query.

Having defined these time-dependent instances, we now show that it suffices to prove gate and query lower bounds for the time-independent Hamiltonians in order to conclude the main results~\Cref{thm:gate,thm:query}:

\begin{lemma}[Time-dependent to time-independent reduction]\label{lemma:time_dep_to_ind_reduction} A gate (query) lower bound of
\begin{align}
    \Omega\lr{\min_{0 \leq K \leq L} \lr{K + \frac{\lambda_K^2}{\epsilon}}}
\end{align}
for the time-independent Hamiltonian in~\Cref{def: time_ind_gate_rra} (\ref{def: time_ind_query_rra}) implies a gate (query) lower bound of 
\begin{align}\label{eq: time_dep_lb_to_show}
    \Omega\lr{\min_{0 \leq K \leq L} \lr{Kt + \frac{t^2\lambda_K^2}{\epsilon}}}
\end{align}
for the time-dependent Hamiltonian in~\Cref{def: time_dep_gate_rra} (\ref{def: time_dep_query_rra}).    
\end{lemma}
\begin{proof}
    The same reduction applies to the gate and query hard instances. In either case, the time-dependent Hamiltonian acts on fresh blocks during consecutive unit-time intervals. More explicitly, for the gate instance one has $h_j^{(b)}=\ketbra{11}_j^{(b)}$ and signs $\theta_j^{(b)}=1$, while for the query instance one has $h_j^{(b)}=\ketbra{1}_j^{(b)}$ and signs $\theta_j^{(b)}=s_j^{(b)}$. Thus the time-dependent Hamiltonian can be written uniformly as
    \begin{align}
        H(\tau)=\sum_{j=1}^L \theta_j^{(b)} a_j h_j^{(b)}, \qquad \tau \in (b-1,b],\quad b\in[t].
    \end{align}
    Since all terms appearing at different times act on disjoint blocks, the ideal time-evolution unitary is
    \begin{align}
        U_H
        &= \prod_{b=1}^t \exp\left(-i\sum_{j=1}^L \theta_j^{(b)}a_jh_j^{(b)}\right)
        = \exp\left(-i\sum_{b=1}^t\sum_{j=1}^L \theta_j^{(b)}a_jh_j^{(b)}\right).
    \end{align}
    Hence the time-dependent instance is equivalent, after relabeling the pair $(b,j)$ as a single index, to a unit-time, time-independent instance with coefficient vector
    \begin{align}
        a^{(t)}=(\underbrace{a_1,\ldots,a_1}_{t},\underbrace{a_2,\ldots,a_2}_{t},\ldots,\underbrace{a_L,\ldots,a_L}_{t}).
    \end{align}
    In the query case, the time-dependent oracle $O_s(b,j)$ is likewise just the oracle for this enlarged time-independent instance under the same relabeling. Let $\lambda_N^{(t)}=\sum_{j>N}a_j^{(t)}$ denote the tail sum of the repeated vector. The assumed unit-time lower bound therefore gives
    \begin{align}
        \Omega\left(\min_{0\leq N\leq Lt}\left(N+\frac{(\lambda_N^{(t)})^2}{\epsilon}\right)\right).
    \end{align}
    It remains to compare this expression to~\cref{eq: time_dep_lb_to_show}. Let $N_*\in\argmin_N(N+(\lambda_N^{(t)})^2/\eps)$. If $N_*=Lt$, the desired inequality is immediate. Otherwise write $N_*=K_*t+r_*$ for $0\leq K_*<L$ and $0\leq r_*<t$, and set $u_*=r_*/t$. Then
    \begin{align}
        \lambda_{N_*}^{(t)}=t\lambda_{K_*+1}+(t-r_*)a_{K_*+1},
    \end{align}
    so
    \begin{align}
        N_*+\frac{(\lambda_{N_*}^{(t)})^2}{\epsilon}
        =K_*t+t\left(u_*+\frac{t}{\eps}((1-u_*)a_{K_*+1}+\lambda_{K_*+1})^2\right).
    \end{align}
    Applying $u+((1-u)x+y)^2\geq \frac14\min\{(x+y)^2,1+y^2\}$ for $u\in[0,1]$ and $x,y\geq0$, with $x=\sqrt{t/\eps}\,a_{K_*+1}$ and $y=\sqrt{t/\eps}\,\lambda_{K_*+1}$, gives
    \begin{align}
        N_* + \frac{(\lambda^{(t)}_{N_*} )^2}{\epsilon} &\ge K_* t + \frac{t}{4} \min\Big( \frac{t}{\eps}(a_{K_* + 1} + \lambda_{K_*+1} )^2, 1 + \frac{t}{\eps} \lambda_{K_* + 1}^2\Big)
        \\
        &\ge \frac{1}{4} \min\Big( K_* t + \frac{t^2 \lambda_{K_*}^2}{\eps}, (K_* + 1) t + \frac{t^2 \lambda_{K_* + 1}^2}{\eps} \Big)
        \\
        &\ge  \frac{1}{4} \min_{0 \leq K \leq L} \Big( Kt + \frac{t^2\lambda_K^2}{\epsilon} \Big),
    \end{align}
    which proves the claimed time-dependent lower bound.
\end{proof}

In both the gate and query lower bounds, we will prove separate ``head'' and ``tail'' lower bounds, capturing the two qualitatively different contributions to the complexity of composite qDRIFT. We provide a brief lemma characterizing the optimal $K_*$ which minimizes the asymptotic cost $K + \lambda_K^2/\eps$ of composite qDRIFT.
\begin{lemma}[Optimal deterministic-randomized split]\label{lemma: optimal_k} Let $K_* \in \argmin_{0 \le K \le L}(K + \lambda_K^2 / \eps)$. Then, $\sum_{j > K_*} a_j^2 \le \eps$ and (if $K_* > 0$) $\eps \le a_{K_*}^2 + 2 a_{K_*} \lambda_{K_*}$.
\end{lemma}
\begin{proof}
    We first show that $\sum_{j > K_*} a_j^2 \le \eps$. If $K_* = L$, this is trivial. Else, by the definition of the minimum, we have
    \begin{align}
        K_* + 1 + \frac{\lambda_{K_*+1}^2}{\eps} \ge K_* + \frac{\lambda_{K_*}^2}{\eps}.
    \end{align}
    After using $\lambda_{K_*+1} = \lambda_{K_*} - a_{K_* + 1}$ and rearranging, we obtain
    \begin{align}
       2 a_{K_* + 1} \lambda_{K_*} -  a_{K_* + 1}^2 \le \eps.
    \end{align}
    We have $a_{K_* + 1} \lambda_{K_*} \le a_{K_* + 1} (2\lambda_{K_*} - a_{K_*+1})$, since $a_{K_*+1} \le \lambda_{K_*}$. Therefore, $a_{K_*+1} \lambda_{K_*} \le \eps$. Finally, since the $a_j$ are sorted in non-increasing order, we have
    \begin{align}
        \sum_{j > K_*} a_j^2 \le a_{K_* + 1} \lambda_{K_*} \le \eps~,
    \end{align}
    as desired. To show that $\eps \le a_{K_*}^2 + 2 a_{K_*} \lambda_{K_*}$, we proceed similarly, but use the other condition on the minimum that $K_* - 1 + \lambda_{K_* - 1}^2 / \eps \ge K_* + \lambda_{K_*}^2 / \eps$.
\end{proof}

\section{Gate lower bound}\label{sec:gate}

In this section, we prove Theorem~\ref{thm:gate}, demonstrating a gate lower bound for Hamiltonian simulation. More specifically, we prove a gate lower bound for any Hamiltonian simulation algorithm which can be expressed as a convex combination of quantum channels whose Stinespring dilations each use at most $G$ two-qubit gates (but may apply arbitrarily many single-qubit gates). As an immediate corollary, this proves a lower bound for mixed unitary channels of bounded two-qubit gate complexity, like those used in Trotterization, QSP, qDRIFT, and composite qDRIFT.

By Lemma~\ref{lemma:time_dep_to_ind_reduction}, it suffices to prove the gate lower bound for unit time evolution under the Hamiltonian $H = H_{2q}(\vec a)$ in Definition~\ref{def: time_ind_gate_rra}, for $\vec a$ a vector of coefficients $\pi/2 \ge a_1 \ge \cdots \ge a_L$. We let $U_H$ denote the corresponding unitary, with channel $\cU_H$. We consider approximating channels $\cE_H$ such that $\frac{1}{2} \| \cU_H(\rho) - \cE_H(\rho) \|_1 \le \eps$ for all density matrices $\rho$, where $\cE_H$ is of the form  $\cE_H = \sum_r p_r \cE_{H,r} = \mathbbm{E}_r [\cE_{H,r}]$, and each $\cE_{H,r}$ is a quantum channel with Stinespring dilation
\begin{align}
    \cE_{H,r}(\rho) = \tr_{\cA_r}(W_r (\rho \otimes \ketbra{0}_{\cA_r}) W_r^\dagger)~.
\end{align}
We assume each $W_r$ contains at most $G$ two-qubit gates on system and ancilla, with arbitrarily many one-qubit gates. 

We begin by exploring the consequences of the bounded gate complexity of the $W_r$. First, we make the simple observation that $W_r$ can only ``touch'' a bounded number of pairs $(A_j,B_j)$ with two-qubit gates. Specifically, let the \emph{touched set} $\tau_r \subseteq [L]$ consist of all $j \in [L]$ such that either $A_j$ or $B_j$ is acted on by at least one two-qubit gate in $W_r$. Since each two-qubit gate acts on a pair of qubits, and every qubit belongs to exactly one pair $(A_j, B_j)$, we must have
\begin{align}
    |\tau_r| \le 2 G.
\end{align}
Second, we observe that only a bounded number of pairs $(A_j, B_j)$ can be ``connected'' by the two-qubit gates in $W_r$. Specifically, let $\cG_r$ denote the graph with vertices corresponding to qubits (both system and ancilla) and edges corresponding to two-qubit gates in $W_r$. Let $\cC(\cG_r)$ denote the set of connected components of $\cG_r$. We define the \emph{connected set} $\sigma_r \subseteq [L]$ as the set of all $j$ such that $A_j$ and $B_j$ belong to the same connected component $C \in \cC(\cG_r)$. Suppose a particular $C$ contains $s$ pairs $(A_j, B_j)$; then, it contains at least $2s$ vertices and therefore at least $2s-1 \ge s$ edges. Summing over connected components, we must have that
\begin{align}
    |\sigma_r| \le G.
\end{align}
Third, we observe that the bounded two-qubit gate complexity of $W_r$ implies a particular factorization of the channel $\cE_{H,r}$. Note that $W_r$ factorizes over the connected components $\cC(\cG_r)$ of $\cG_r$; $W_r = \otimes_{C \in \cC(\cG_r)} W_{r,C}$. Moreover, since the ancilla qubits are initially in a product state and the partial trace also factorizes, the channel $\cE_{H,r}$ itself factorizes; $\cE_{H,r} = \otimes_{C \in \cC(\cG_r)} \cE_{H,r, C}$. 

We now briefly sketch how the preceding consequences of bounded two-qubit gate complexity lead to the desired lower bounds. As discussed in~\Cref{appendix: prelims}, we prove separate ``head'' and ``tail'' lower bounds of $G \ge \Omega(K_*)$ and $G \ge \Omega(\lambda_{K_*}^2/\eps)$, where $K_*$ minimizes the asymptotic cost $K + \lambda_K^2/\eps$ of composite qDRIFT. Both bounds utilize the factorization of the channels $\cE_{H,r}$, which is in tension with the correlations created by the ideal evolution $\cU_H$ across many pairs of qubits $(A_j, B_j)$. We detect this tension using complementary witness states for the head and tail lower bounds. For the head bound, we simply use the product state $\ket{+}^{\otimes 2L}$. The large head coefficients mean that the ideal evolution entangles many pairs $(A_j, B_j)$. By contrast, if $W_r$ uses too few two-qubit gates, then many such pairs remain disconnected (here we use the bounded size of the connected set), and one can find a bipartition across which many pairs $(A_j, B_j)$ remain in a product state. Its fidelity with the ideal, entangled state decreases exponentially in the number of separated pairs, yielding $G \ge \Omega(K_*)$.

For the tail, the coefficients are individually too weak to repeat the above argument. However, as discussed in the main text, by using cat state witnesses we can coherently add many tail coefficients in the accumulated relative phase. More precisely, we consider cat states which superpose the all-zero state with states that have ones in some subset of the tail coefficients, on either the $A$ qubits, the $B$ qubits, or both. For the ideal evolution, only the final state acquires a relative phase. However, for the approximate evolution, the former states may also acquire a relative phase, breaking the perfect correlation between pairs $(A_j, B_j)$. The number of possible violations increases as the size of the touched set decreases. We witness these violations by considering random cat states of the above form; combined with the bound on the size of the touched set we derive $G \ge \Omega(\lambda_{K_*}^2/\eps)$.

We now provide formal proofs of the two lower bounds, beginning with the tail bound.

\begin{lemma}[Gate complexity tail lower bound]~\label{lemma: gate_tail_lb} Let $\eps \in (0, \eps_0)$ for sufficiently small constant $\eps_0 > 0$.
Let $\cU_H$ denote time evolution for unit time under $H = \sum_{j=1}^L a_j \ketbra{11}_j$, and let $\cE_H = \sum_r p_r \cE_{H,r}$, where each $\cE_{H,r}$ has a Stinespring dilation with at most $G$ two-qubit gates. Suppose $\frac{1}{2} \| \cU_H(\rho) - \cE_H(\rho) \|_1 \le \eps$ for all density matrices $\rho$. Let $K_* \in \argmin_{0 \le K \le L}(K + \lambda_K^2 / \eps)$. Then 
\begin{align}
    G + 1 \ge C_t \frac{\lambda_{K_*}^2}{\eps}
\end{align}
for some constant $C_t > 0$.
\end{lemma}

\begin{proof}
    For $x, y \in \{0,1\}^{L}$, let $|x,y\rangle = \ket{x_1}_{A_1} \otimes \cdots \otimes \ket{x_L}_{A_L} \otimes \ket{y_1}_{B_1} \otimes \cdots \otimes \ket{y_L}_{B_L}$. Let $h(x,y) = \sum_{j=1}^L a_j x_j y_j$ and $f(x,y) = \exp(-i h(x,y))$, so that $U_H \ket{x,y} = f(x,y) \ket{x,y}$. Similarly, for each branch $r$ of the approximating channel, let $g_r(x,y) = \bra{x,y} \cE_{H,r} ( \ketbra{x,y}{0,0} ) \ket{0,0}$; note that $|g_r(x,y)| \le 1$ and that $g_r(x,y)$ also factorizes over connected components of $\cG_r$; $g_r(x,y) = \prod_{C \in \cC(\cG_r)} g_{r,C}(x_C,y_C)$, where $x_C$ and $y_C$ denote $x$ and $y$ restricted to $C$, respectively. Finally, let $g(x,y) = \sum_r p_r g_r(x,y)$. We claim that $f(x,y)$ and $g(x,y)$ must be close for every $x$ and $y$, since $\cE_H$ gives a good approximation to $\cU_H$. To see this, for each $(x,y) \neq (0,0)$ and $c \in \{0,1,2,3\}$, define $\ket{\psi_c(x,y)} = (\ket{0,0} + i^c \ket{x,y} ) / \sqrt{2}$ and $\psi_c(x,y) = \ketbra{\psi_c(x,y)}$. Observe that $\ketbra{x,y}{0,0} = \frac{1}{2} \sum_{c=0}^3 i^{-c} \psi_c(x,y)$; hence
\begin{align}
    |f(x,y) - g(x,y)| &= | \tr(\ketbra{0,0}{x,y} \left( \cU_H(\ketbra{x,y}{0,0})-\cE_H(\ketbra{x,y}{0,0}) \right) ) |
    \\
    &\le \| \cU_H(\ketbra{x,y}{0,0}) - \cE_H(\ketbra{x,y}{0,0}) \|_1
    \\
    &\le 4 \eps~, \label{eq:off_diag_exp_bound}
\end{align}
where the first inequality is H\"{o}lder's and the second is triangle inequality plus $\| \cU_H(\psi_c(x,y))-\cE_H(\psi_c(x,y))  \|_1 \le 2 \eps$ for each $c$. It follows similarly that $|f(0,0) - g(0,0) | \le 2\eps < 4 \eps$ by considering the state $\ketbra{0,0}$ directly.

We have shown $|f(x,y) - g(x,y)| = |\mathbbm{E}_r[f(x,y) - g_r(x,y)]| \le 4 \eps$; we claim also that
\begin{align}\label{eq: off_diag_second_moment_bound}
    \mathbbm{E}_r[|f(x,y) - g_r(x,y)|^2] \le 8 \eps.
\end{align}
To see this, observe that 
\begin{align}
    \mathbbm{E}_r[|f(x,y) - g_r(x,y)|^2] &= \mathbbm{E}_r[|f(x,y)|^2 + |g_r(x,y)|^2 - 2 \mathrm{Re}[f^*(x,y) g_r(x,y)]]\label{eq:adfirst}
    \\
    &\le 2 \mathbbm{E}_r[ 1 - \mathrm{Re}[f^*(x,y) g_r(x,y)]]
    \\
    &\le 2 |1 - \mathbbm{E}_r[f^*(x,y) g_r(x,y)]|
    \\
    &= 2 |f(x,y) - \mathbbm{E}_r[g_r(x,y)]|
    \\
    &\le 8 \eps.\label{eq:adlast}
\end{align}
Here, the first inequality holds since $|f(x,y)| , |g_r(x,y)| \le 1$, the second by upper bounding real part by absolute value, and the third by \cref{eq:off_diag_exp_bound}.

    Let $T = \{K_* +1, \dots, L\}$ denote the tail set. For each subset $S \subseteq T$, let $\mathbbm{1}_S$ denote the $L$-bit string with ones at the locations in $S$ and zeroes elsewhere. We define four $2L$-bit strings: $z_S^{(00)} = (0,0)$, $z_S^{(01)} = (0, \mathbbm{1}_S)$, $z_S^{(10)} = (\mathbbm{1}_S, 0)$, and $z_S^{(11)} = (\mathbbm{1}_S, \mathbbm{1}_S)$. We have
    \begin{align}
        f(z_S^{(00)}) = f(z_S^{(01)}) = f(z_S^{(10)}) = 1 \quad \text{and} \quad f(z_S^{(11)}) = e^{-i \sum_{j \in S} a_j}.
    \end{align}
    For each branch $r$ and set $S \subseteq T$, we define 
    \begin{align}
        \widetilde{R}_r(S) = \frac{g_r(z_S^{(00)})g_r(z_S^{(11)})}{g_r(z_S^{(01)})g_r(z_S^{(10)})}
    \end{align}
    when the denominator $g_r(z_S^{(01)})g_r(z_S^{(10)}) \neq 0$; else let $\widetilde{R}_r(S) = 0$. Finally, define $R_r(S)$ as $\widetilde{R}_r(S)$ projected onto the unit disk, so that $|R_r(S)| \le 1$. We call $R_r(S)$ the \emph{subset correlation ratio} of $S$ on branch $r$. The key property of $R_r(S)$ is that (when the denominator is nonzero), we have
    \begin{align}\label{eq: subset_corr_condition}
        R_r(S) = R_r(S \cap \tau_r)
    \end{align}
    for $\tau_r$ the touched set of branch $r$. To see why, note that if $j \notin \tau_r$, then both $A_j$ and $B_j$ belong to singleton connected components of $\cG_r$, and the factorization formula for $g_r$ implies that the contribution to $R_r(S)$ from $j$ cancels in numerator and denominator. 
    
    We claim that, for \emph{any} distribution over subsets $S \subseteq T$, the subset correlation ratio $R_r(S \cap \tau_r)$ actually \emph{concentrates} to $e^{-i \sum_{j \in S} a_j}$ (note that this is the value of the subset correlation ratio if $g_r$ is replaced by $f$). Quantitatively, we claim that
    \begin{align}\label{eq: corr_ratio_conc}
        \mathbbm{E}_{r,S}[| R_r(S \cap \tau_r) - e^{-i \sum_{j \in S} a_j}|] \le 99 \sqrt{\eps}.
    \end{align}
    To prove the claim, first let $\bm{A}$ denote the event that $|g_r(z_S^{(01)})-1|$ and $|g_r(z_S^{(10)})-1|$ are both at most $\frac{1}{2}$. We have
    \begin{align}
        \mathbbm{E}_{r,S}[| R_r(S \cap \tau_r) - e^{-i \sum_{j \in S} a_j}|] \le \mathbbm{E}_{r,S}[\lvert R_r(S \cap \tau_r) - e^{-i \sum_{j \in S} a_j}\rvert \mid \bm{A}] \mathbbm{P}[\bm{A}] + 2 \mathbbm{P}[\bm{A}^c] 
    \end{align}
    using the trivial bound $| R_r(S \cap \tau_r) - e^{-i \sum_{j \in S} a_j}| \le 2$. For the second term, by union bound and Markov's inequality, we have 
    \begin{align}\label{eq: ratio_conc_second_term}
        2\mathbbm{P}[\bm{A}^c] \le 8 \mathbbm{E}_{r,S}[|g_r(z_S^{(01)})-1|^2 + |g_r(z_S^{(10)})-1|^2].
    \end{align}
    For the first term, conditional on $\bm{A}$ we have $|g_r(z_S^{(01)})g_r(z_S^{(10)})| \ge \frac{1}{4}$, so in particular the denominator of $\widetilde{R}_r(S)$ is nonzero, and we may use~\cref{eq: subset_corr_condition} to conclude
    \begin{align}
        \mathbbm{E}_{r,S}[\lvert R_r(S \cap \tau_r) - e^{-i \sum_{j \in S} a_j}\rvert \mid \bm{A}] &= \mathbbm{E}_{r,S}[\lvert R_r(S) - e^{-i \sum_{j \in S} a_j}\rvert \mid \bm{A}]
        \\
        &\le \mathbbm{E}_{r,S}[\lvert \widetilde{R}_r(S) - e^{-i \sum_{j \in S} a_j}\rvert \mid \bm{A}]~,
    \end{align}
    where the inequality follows since $R_r(S)$ is the projection of $\widetilde{R}_r(S)$ onto the unit disk. Now, by repeated application of triangle inequality along with $|g_r| \le 1$, we have conditional on $\bm{A}$ that $\lvert \widetilde{R}_r(S) - e^{-i \sum_{j \in S} a_j}\rvert  \le 4(| g_r(z_S^{(11)}) -  e^{-i \sum_{j \in S} a_j}| + | g_r(z_S^{(00)}) - 1| + | g_r(z_S^{(01)}) - 1| + | g_r(z_S^{(10)}) - 1|)$, where the factor of 4 comes from the minimum modulus of the denominator of $\widetilde{R}_r(S)$. 
    Applying Cauchy-Schwarz and Jensen's inequality, we have
    \begin{align}\label{eq: ratio_conc_first_term}
        \mathbbm{E}_{r,S}[\lvert R_r(S \cap \tau_r) - e^{-i \sum_{j \in S} a_j}\rvert \mid \bm{A}] \mathbbm{P}[\bm{A}] \le 8 \Big(  \mathbbm{E}_{r,S}[ | g_r(z_S^{(11)}) -  e^{-i \sum_{j \in S} a_j}|^2 &+ | g_r(z_S^{(00)}) - 1|^2 
        \\
        &+ | g_r(z_S^{(01)}) - 1|^2 + | g_r(z_S^{(10)}) - 1|^2] \Big)^{1/2} \nonumber
    \end{align}
    Combining eqs.~\eqref{eq: ratio_conc_second_term} and~\eqref{eq: ratio_conc_first_term} and repeatedly applying~\cref{eq: off_diag_second_moment_bound}, we conclude that
    \begin{align}
        \mathbbm{E}_{r,S}[| R_r(S \cap \tau_r) - e^{-i \sum_{j \in S} a_j}|] \le 8 \sqrt{32 \eps} + 128 \eps \le 99 \sqrt{\eps}~,
    \end{align}
    as desired, where the final inequality holds for $\eps \le \frac{1}{10}$.

    Eq.~\eqref{eq: corr_ratio_conc} is true for any distribution over $S \subseteq T$. To conclude the proof, we will choose two distributions $\cP_+$ and $\cP_-$ over $S \subseteq T$ such that (i) the values to which $\textrm{exp}(-i \sum_{j \in S} a_j)$ concentrates under the two distributions are sufficiently far apart but (ii) $R_r(S \cap \tau_r)$ remains indistinguishable under the two distributions, unless the number of two-qubit gates $G$ is sufficiently large. Together with~\cref{eq: corr_ratio_conc}, this proves the gate lower bound. We define $\cP_{\pm}$ as the distribution over $S \subseteq T$ such that every $j \in T$ is included independently with biased probability $p_{\pm} = (1 \pm \eta)/2$, for $\eta = 1000 \sqrt{\eps} /\lambda_{K_*} = 1000 \sqrt{\eps} / (\sum_{j \in T} a_j)$. (We may assume $\lambda_{K_*} \neq 0$, else the tail lower bound is trivial. We may also assume $\eta \le 1/2$ by assuming $\lambda_{K_*}/\sqrt{\eps} \ge 2000$ and choosing the constant $C_t$ to be at most $2.5 \times 10^{-7}$ to handle the other case.) First, we observe that
    \begin{align}\label{eq: conc_value_bound}
        \big| e^{-i p_+ \sum_{j \in T} a_j} - e^{-i p_- \sum_{j \in T} a_j} \big| = 2 \hspace{.05cm}\mathrm{sin}\Big(\frac{\eta}{2} \sum_{j \in T} a_j\Big) \ge 500 \sqrt{\eps}
    \end{align}
    where the final inequality holds for $\eps \le \eps _0 = 10^{-5}$. On the other hand, we have
    \begin{align}
        \mathbbm{E}_{r, S \sim \cP_{\pm}}[ |R_r(S \cap \tau_r) - e^{-i p_{\pm} \sum_{j \in T} a_j} | ] \le 
        \mathbbm{E}_{r, S \sim \cP_{\pm}}[ |R_r(S \cap \tau_r) - &e^{-i \sum_{j \in S} a_j} | ] 
        \\
        &+ \mathbbm{E}_{S \sim \cP_{\pm}}[ | e^{-i \sum_{j \in S} a_j} - e^{-i p_{\pm} \sum_{j \in T} a_j} | ] \nonumber 
    \end{align}
    The first term is at most $99 \sqrt{\eps}$ by~\cref{eq: corr_ratio_conc}; the second term is bounded as
    \begin{align}
        \mathbbm{E}_{S \sim \cP_{\pm}}[ | e^{-i \sum_{j \in S} a_j} - e^{-i p_{\pm} \sum_{j \in T} a_j} | ] &\le \mathbbm{E}_{S \sim \cP_{\pm}}\Big[ \Big| \sum_{j \in S} a_j - p_{\pm} \sum_{j \in T} a_j \Big| \Big]
        \\
        &= \mathbbm{E}_{S \sim \cP_{\pm}}\Big[ \Big| \sum_{j \in S} a_j -  \mathbbm{E}_{S' \sim \cP_{\pm}}\Big[ \sum_{j \in S'} a_j \Big] \Big| \Big]
        \\
        &\le \Big( \Var_{S \sim \cP_{\pm}}\Big[ \sum_{j \in S} a_j \Big] \Big)^{1/2}
        \\
        &= \Big( p_{\pm} (1 - p_{\pm}) \sum_{j \in T} a_j^2\Big)^{1/2}
        \\
        &\le \sqrt{\eps}/2.
    \end{align}
    Here, the first inequality is $|e^{ix} - e^{iy}| \le |x-y|$, the second is Jensen's, and the third follows from~\Cref{lemma: optimal_k} and $p_{\pm} (1 - p_{\pm}) \le 1/4$. Combining this with~\cref{eq: corr_ratio_conc}, we conclude that
    \begin{align}\label{eq: better_conc}
        \mathbbm{E}_{r, S \sim \cP_{\pm}}[ |R_r(S \cap \tau_r) - e^{-i p_{\pm} \sum_{j \in T} a_j} | ] \le 100 \sqrt{\eps}.
    \end{align}
    
    Now, finally, we compare the total variation distance between the distributions of $R_r(S \cap \tau_r)$ under $(r, S) \sim (p_r, \cP_{\pm})$; call these two distributions $p_{R,\pm}$. By data processing, the total variation distance between $p_{R,\pm}$ is at most the total variation distance between the distributions of $(r, S \cap \tau_r)$ under $(r, S) \sim (p_r, \cP_{\pm})$; call these two distributions $p_{\cap, \pm}$. By Pinsker's inequality, we may relate the total variation distance $\mathrm{TV}(p_{\cap, +},p_{\cap, -})$ to the KL divergence $D_{\mathrm{KL}}(p_{\cap, +} \| p_{\cap, -})$ as follows:
    \begin{align}
        \mathrm{TV}(p_{\cap, +},p_{\cap, -}) \le \Big( \frac{1}{2} D_{\mathrm{KL}}(p_{\cap, +} \| p_{\cap, -}) \Big)^{1/2}.
    \end{align}
    By the chain rule of KL divergence, $D_{\mathrm{KL}}(p_{\cap, +} \| p_{\cap, -})$ is the conditional KL divergence between the conditional distributions of $S \cap \tau_r$ given $r$. The conditional distribution is just a product of $|\tau_r|$ Bernoullis with probability $p_{\pm}$, so using $D_{\mathrm{KL}}(\bern(p_+) \|  \bern(p_-)) = 2\eta \tanh^{-1}(\eta)$
    
    we have
    \begin{align}
         \mathrm{TV}(p_{\cap, +},p_{\cap, -}) \le \Big( \frac{1}{2} \sum_r p_r | \tau_r| 2 \eta \tanh^{-1}(\eta)\Big)^{1/2} \le \Big( \frac{3 \eta^2}{2} \sum_r p_r | \tau_r|\Big)^{1/2} ~,
    \end{align}
    where the second inequality holds since $\eta \le 1/2$. Finally, $|\tau_r| \le 2G$, so we have
    \begin{align}\label{eq: TV_bound}
        \mathrm{TV}(p_{R, +},p_{R, -}) \le \mathrm{TV}(p_{\cap, +},p_{\cap, -}) \le (3 \eta^2 G)^{1/2}.
    \end{align}
    We show that this leads to a lower bound on the number of two-qubit gates $G$. Consider any $y$ in the support of $p_{R, \pm}$; we have
    \begin{align}
         \big| e^{-i p_+ \sum_{j \in T} a_j} - e^{-i p_- \sum_{j \in T} a_j} \big| \le  \big| e^{-i p_+ \sum_{j \in T} a_j} - y \big| +  \big| y - e^{-i p_- \sum_{j \in T} a_j} \big|.
    \end{align}
    Multiply both sides by $\min(p_{R, +}(y), p_{R, -}(y))$ and sum over $y$; by the overlap characterization of total variation distance $\sum_y \min(p_{R, +}(y), p_{R, -}(y)) = 1 - \mathrm{TV}(p_{R, +},p_{R, -})$, we have
    \begin{align}
        (1 - \mathrm{TV}(p_{R, +},p_{R, -}))\big| e^{-i p_+ \sum_{j \in T} a_j} - e^{-i p_- \sum_{j \in T} a_j} \big| \le \mathbbm{E}_{r, S \sim \cP_+}\big[\big| R_r(S &\cap \tau_r) - e^{-i p_+ \sum_{j \in T} a_j} \big|\big] \label{eq:tvovp}
        \\
        &+ \mathbbm{E}_{r, S \sim \cP_-}\big[\big| R_r(S \cap \tau_r) - e^{-i p_- \sum_{j \in T} a_j} \big|\big] \nonumber
    \end{align}
    By~\cref{eq: conc_value_bound} and~\cref{eq: TV_bound}, the left-hand side is at least $500 (1 - (3 \eta^2 G)^{1/2})\sqrt{\eps}$, and by~\eqref{eq: better_conc} the right-hand side is at most $200 \sqrt{\eps}$; it follows after rearranging and substituting $\eta = 1000 \sqrt{\eps} / \lambda_{K_*}$ that
    \begin{align}
        G \ge C_t \frac{\lambda_{K_*}^2}{\eps}
    \end{align}
    for $C_t = 10^{-7}$. This proves the tail lower bound.
\end{proof}

Next, we prove the head lower bound.

\begin{lemma}[Gate complexity head lower bound]~\label{lemma: gate_head_lb} Let $\eps \in (0, \eps_0)$ for sufficiently small constant $\eps_0 > 0$.
Let $\cU_H$ denote time evolution for unit time under $H = \sum_{j=1}^L a_j \ketbra{11}_j$, and let $\cE_H = \sum_r p_r \cE_{H,r}$, where each $\cE_{H,r}$ has a Stinespring dilation with at most $G$ two-qubit gates. Suppose $\frac{1}{2} \| \cU_H(\rho) - \cE_H(\rho) \|_1 \le \eps$ for all density matrices $\rho$. Let $K_* \in \argmin_{0 \le K \le L}(K + \lambda_K^2 / \eps)$. Then 
\begin{align}
    G + C_{h,1} \ge C_{h,2} K_* - C_{h,3} \frac{\lambda_{K_*}^2}{\eps}
\end{align}
for some constants $C_{h,1}, C_{h,2}, C_{h,3} > 0$.
\end{lemma}

\begin{proof}
    Let $\ket{\psi} = \ket{+}^{\otimes 2L}$ and $\psi = \ketbra{\psi}$. Let $\psi_{\cU_H} = \cU_H(\psi)$ and $\psi_{\cE_{H,r}} = \cE_{H,r}(\psi)$. By assumption, we have $\frac{1}{2} \| \psi_{\cU_H} - \sum_r p_r \psi_{\cE_{H,r}} \|_1 \le \eps$. By Fuchs van de Graaf, this implies that
    \begin{align}\label{eq: plus_state_bound}
        \mathbbm{E}_r[ \bra{\psi_{\cU_H}} \psi_{\cE_{H,r}} \ket{\psi_{\cU_H}}] \ge 1-\eps.
    \end{align}
    We must upper bound each $\bra{\psi_{\cU_H}} \psi_{\cE_{H,r}} \ket{\psi_{\cU_H}}$. Recall that $\sigma_r \subseteq [L]$ denotes the set of $j$ such that $A_j$ and $B_j$ belong to the same connected component of $\cG_r$. We claim that, for each $r$,
    \begin{align}
        \bra{\psi_{\cU_H}} \psi_{\cE_{H,r}} \ket{\psi_{\cU_H}} \le e^{-\frac{1}{32} \sum_{j \notin \sigma_r} a_j^2}~.
    \end{align}
    To see this, choose a random bipartition of the connected components $\cC(\cG_r)$ of $\cG_r$. For each $j \notin \sigma_r$, $A_j$ and $B_j$ are \emph{separated} by the bipartition with probability $1/2$; hence, there exists a bipartition $\cB \subset \cC(\cG_r)$ with
    \begin{align}\label{eq: random_bipartition}
        \sum_{j:~A_j,B_j \text{ separated by } \cB} a_j^2 \ge \frac{1}{2} \sum_{j \notin \sigma_r} a_j^2~.
    \end{align}
    $\cE_{H,r}$ factorizes over connected components; hence $\psi_{\cE_{H,r}}$ is a product state across the bipartition $\cB$. Therefore $\bra{\psi_{\cU_H}} \psi_{\cE_{H,r}} \ket{\psi_{\cU_H}}$ is at most the largest squared Schmidt coefficient of $\ket{\psi_{\cU_H}}$ across $\cB$. To bound this largest squared Schmidt coefficient, first observe that $\ket{\psi_{\cU_H}}$ factorizes as
    \begin{align}
        \ket{\psi_{\cU_H}} = \bigotimes_{j = 1}^L e^{-i a_j \ketbra{11}_j } \ket{+}_{A_j} \ket{+}_{B_j}.
    \end{align}
    The tensor product of all factors such that $A_j$ and $B_j$ lie in the same half of the bipartition is a product state across $\cB$. Therefore, the Schmidt coefficients of $\ket{\psi_{\cU_H}} $ are simply the products of all Schmidt coefficients of the individual factors $e^{-i a_j \ketbra{11}_j } \ket{+}_{A_j} \ket{+}_{B_j}$ such that $A_j$ and $B_j$ are separated by $\cB$. To compute the maximum squared Schmidt coefficient, it therefore suffices to focus on a single one of these. By an explicit computation, the maximum squared Schmidt coefficient of any such factor is $\cos^2(a_j/4)$. Therefore, we have
    \begin{align}
        \bra{\psi_{\cU_H}} \psi_{\cE_{H,r}} \ket{\psi_{\cU_H}} \le \prod_{j:~A_j,B_j \text{ separated by } \cB} \cos^2(\frac{a_j}{4}).
    \end{align}
    Using $\cos(x) \le e^{-x^2/2}$ for $|x| \le \pi/2$ as well as~\cref{eq: random_bipartition}, we have
    \begin{align}
        \bra{\psi_{\cU_H}} \psi_{\cE_{H,r}} \ket{\psi_{\cU_H}} \le e^{-\frac{1}{32} \sum_{j \notin \sigma_r} a_j^2},
    \end{align}
    as desired. From~\cref{eq: plus_state_bound}, we have
    \begin{align}
        \mathbbm{E}_r[ 1 - e^{-\frac{1}{32} \sum_{j \notin \sigma_r} a_j^2}] \le \eps.
    \end{align}
    Let $d_r = |\{ j \le K_* : j \notin \sigma_r\}|$ be the number of disconnected pairs $A_j,B_j$ in the ``head'' $j \le K_*$.  Since $a_j \ge a_{K_*}$ for all $j \le K_*$, we have
    \begin{align}
        \mathbbm{E}_r[ 1 - e^{-\frac{1}{32} d_r a_{K_*}^2}] \le \eps.
    \end{align}
    Applying $1 - e^{-x} \ge \frac{1}{2} \min(1,x)$, we conclude that $\mathbbm{E}_r[ \min(1, \frac{1}{32} d_r a_{K_*}^2)] \le 2 \eps$. Next, for each $d_r \le K_*$, we have $d_r \le \frac{32}{a_{K_*}^2} \min(1, \frac{1}{32} d_r a_{K_*}^2) + K_* \mathbbm{1}_{a_{K_*}^2 d_r / 32 > 1}$; taking the expected value of both sides and applying Markov's inequality, we conclude that
    \begin{align}
        \mathbbm{E}_r[d_r] \le \frac{64 \eps}{a_{K_*}^2} + 2 \eps K_*.
    \end{align}
    Now, $d_r = K_* - | \sigma_r \cap [K_*] |$, so we have $\mathbbm{E}_r[| \sigma_r \cap [K_*] |] \ge (1 - 2 \eps) K_* - 64 \eps / a_{K_*}^2$. Since $| \sigma_r \cap [K_*] | \le |\sigma_r| \le G$, this implies that
    \begin{align}
        G \ge (1 - 2 \eps) K_* - \frac{64 \eps}{a_{K_*}^2} \ge \frac{K_*}{2} - \frac{64 \eps}{a_{K_*}^2},
    \end{align}
    for $\eps \le 1/4$. We claim that $\eps / a_{K_*}^2 \le 2 + 16 \lambda_{K_*}^2 / \eps$. There are two cases to consider. If $\eps / a_{K_*}^2 \le 2$, we are done.  Else, $\eps / a_{K_*}^2 > 2$, and by Lemma~\ref{lemma: optimal_k} we have $2 a_{K_*} \lambda_{K_*} \ge \eps - a_{K_*}^2 > \eps / 2$ (we may assume $K_* > 0$, else the head lower bound trivially holds). Hence, $\eps / a_{K_*}^2 \le \eps (4 \lambda_{K_*}/\eps)^2 = 16 \lambda_{K_*}^2 / \eps$, and we are done. It follows that
    \begin{align}
         G + C_{h,1} \ge C_{h,2} K_* - C_{h,3} \frac{\lambda_{K_*}^2}{\eps}
    \end{align}
    for $C_{h,1} = 128$, $C_{h,2} = 1/2$, and $C_{h,3} = 1024$. This proves the head lower bound.
\end{proof}

Combining Lemmas~\ref{lemma: gate_tail_lb} and~\ref{lemma: gate_head_lb} (as well as the reduction Lemma~\ref{lemma:time_dep_to_ind_reduction} from time-dependent to time-independent lower bounds) proves Theorem~\ref{thm:gate}.

\section{Query lower bound}\label{sec:query}

In this section, we prove~\Cref{thm:query}, demonstrating a classical query lower bound for Hamiltonian simulation. The query lower bound follows largely the same proof strategy as the gate lower bound, but it is somewhat simpler since it suffices to consider a 1-local Hamiltonian family. As before, we first prove a unit-time lower bound, which extends to asymptotically late times via \Cref{lemma:time_dep_to_ind_reduction}. Recall that our hard Hamiltonians, defined in~\Cref{def: time_ind_query_rra}, are of the form 
\begin{align}
	H_s = \sum_{j=1}^L s_j a_j \ketbra{1}_j, \qquad s \in \{\pm 1\}^L
\end{align}
and that we have query access to $H$ via a classical oracle $O_H$ that returns $s_j$ and $\ketbra{1}_j$ on input $j$ (\Cref{def: oracle}). Since the operators in the Hamiltonian are fixed, and to make the $s$-dependence explicit, we will refer to the oracle below as $O_s$, which takes $j$ and returns $s_j$.
As before, we will use $r$ to denote the transcript of the algorithm, and $\tau_r\subseteq [L]$ to denote the set of distinct queried indices in the transcript $r$. Conditional on $r$, the algorithm is a channel that depends only on the transcript (i.e., the signs that were queried). 

We proceed by showing unit-time tail and head lower bounds analogously to the gate lower bounds. The intuition for the proof technique is similar; at a technical level, the tail bounds are proven by choosing a parametrized bias (either towards $+1$ or $-1$) and then sampling the signs independently from a biased Bernoulli distribution. For the head bounds, it suffices to sample the signs uniformly at random. Explicitly, in \Cref{lemma: query_tail_lb} and \Cref{lemma: query_head_lb} we will define integer $K_* \in \argmin_{0 \le K \le L}(K + \lambda_K^2 / \eps)$ and the distribution $\cP = (\cP_+ + \cP_-)/2$ in terms of product distributions
\begin{align}\label{eq:prs}
    \mathbbm{P}_{\cP_\pm}[s_j = +1] = \begin{cases}1/2 & j \leq K_* \\ (1\pm \eta)/2 > 0 & j > K_* \end{cases} \qquad \text{for} \qquad \eta = \min\left\{\frac{16\sqrt \epsilon}{\lambda_{K_*}}, \frac{1}{2}\right\},
\end{align}
where we use the convention $\eta=1/2$ when $\lambda_{K_*}=0$.

\begin{lemma}[Query complexity tail lower bound]\label{lemma: query_tail_lb}
Let $\epsilon \in (0, \epsilon_0)$ for sufficiently small constant $\epsilon_0 > 0$. Let $\cU_s$ denote time evolution for unit time under $H_s = \sum_{j=1}^L s_j a_j \ketbra{1}_j$. For $K_* \in \argmin_{0 \le K \le L}(K + \lambda_K^2 / \eps)$ and $\cP$ defined in \cref{eq:prs}, suppose an algorithm with classical oracle access to $O_s$ outputs a channel $\cE_s$ satisfying $\frac{1}{2}\EE_{s\sim\cP}\norm{\cU_s(\rho) - \cE_s(\rho)}_1 \leq \epsilon$ for all density matrices $\rho$. Then the algorithm must make $Q$ queries to $O_s$ for $Q$ satisfying
\begin{align}
	Q + 1 \geq c_t \frac{\lambda_{K_*}^2}{\epsilon}
\end{align}
for some constant $c_t > 0$.
\end{lemma}
\begin{proof}
	If $\lambda_{K_*}=0$, the claimed tail bound is trivial, so assume $\lambda_{K_*}>0$. If $\lambda_{K_*}\leq 32\sqrt{\epsilon}$, then $Q+1 \geq \frac{1}{1024}\frac{\lambda_{K_*}^2}{\epsilon}$, so the result holds with constant $c_t=1/1536$. Hence, we can assume for the remainder of the proof that $\lambda_{K_*}>32\sqrt{\epsilon}$ and $\eta=\frac{16\sqrt{\epsilon}}{\lambda_{K_*}}$.
    
	For state $| 1_{\rm tail} \rangle = | 0^{K_*}1^{L-K_*} \rangle$, we consider ideal time evolution
	\begin{align}\label{eq:idealpsi}
		e^{-iH_s} \lr{\frac{| 0^L \rangle  + | 1_{\rm tail} \rangle}{\sqrt 2}} = \frac{| 0^L \rangle + \mathrm{exp}\big(-i\sum_{j=K_*+1}^L s_j a_j\big) | 1_{\rm tail} \rangle}{\sqrt 2}.
	\end{align}
	If too few signs are queried, a dephasing error occurs. To control this dephasing error, it suffices --- as in the gate lower bound --- to isolate the off-diagonal in the $\{ | 0^L \rangle, | 1_{\rm tail} \rangle \}$ subspace, since \cref{eq:idealpsi} lies in this subspace, and to compare it with the ideal time evolution
	\begin{align}\label{eq:idealphase}
		\cU_s\lr{\ketbra{1_{\rm tail}}{0^L}} = \mathrm{exp}\Big(-i\sum_{j=K_*+1}^L s_j a_j\Big)\ketbra{1_{\rm tail}}{0^L}.
	\end{align}
	When $s$ is drawn according to $\cP_\pm$ and $r$ is the resulting transcript, let $\mu_\pm$ be the distribution of the random variable $\bra{1_{\rm tail}}\cE_r\lr{\ketbra{1_{\rm tail}}{0^L}}\ket{0^L}$, and let $\cR_\pm$ be the transcript distributions. To show a query lower bound, we place upper and lower bounds on the total variation distance $\mathrm{TV}(\mu_+ , \mu_-)$; the upper bound will depend on $Q$, while the lower bound will be independent. By the data processing inequality and Pinsker's inequality,
	\begin{align}
		\mathrm{TV}(\mu_+ , \mu_-) \leq \mathrm{TV}(\cR_+ , \cR_-) \leq \left( \frac{1}{2} D_\KL(\cR_+ \| \cR_-) \right)^{1/2}.
	\end{align}
    As in the gate lower bound, by chain rule, each query to a fresh tail index increases the KL divergence by $2\eta \tanh^{-1}(\eta)$, 
	while head queries and repeated tail queries contribute 0. We use the bound $2 \eta \tanh^{-1}(\eta) \leq 3 \eta^2$ for $0\leq \eta \leq 1/2$ to obtain
	\begin{align}\label{eq:tvub}
		\mathrm{TV}(\mu_+ , \mu_-) \leq \eta\sqrt{\frac{3Q}{2}}.
	\end{align}
	We now wish to lower bound this by a $Q$-independent quantity. We lower bound the total variation distance by measuring how far it is from ideal evolution \cref{eq:idealphase} under the expected phase. That is, for expected phases
	\begin{align}
		\phi_\pm = \mathbbm{E}_{\cP_\pm}\Big[\sum_{j=K_*+1}^L s_j a_j \Big] = \pm \eta\lambda_{K_*},
	\end{align}
	we have similarly to \cref{eq:tvovp} that
	\begin{align}
		\mathrm{TV}(\mu_+ , \mu_-) \geq 1 - \frac{\EE_{Y\sim\mu_+} \big[\abs{Y - e^{-i\phi_+}}\big] + \EE_{Y\sim\mu_-} \big[\abs{Y - e^{-i\phi_-}}\big]}{\abs{e^{-i\phi_+} - e^{-i\phi_-}}}.
	\end{align}
	The denominator is readily bounded by
	\begin{align}
		\abs{e^{-i\phi_+} - e^{-i\phi_-}} = \abs{e^{-i\eta\lambda_{K_*}} - e^{i\eta\lambda_{K_*}}} = 2\sin(16\sqrt\epsilon) \geq 16\sqrt\epsilon,
	\end{align}
	for sufficiently small $\epsilon < \epsilon_0$. We claim that the numerator is upper-bounded by
	\begin{align}\label{eq:numbound}
		\EE_{Y\sim\mu_+} \big[\abs{Y - e^{-i\phi_+}}\big] + \EE_{Y\sim\mu_-} \big[\abs{Y - e^{-i\phi_-}}\big] \leq 8 \sqrt\epsilon,
	\end{align}
	giving $\mathrm{TV}(\mu_+ , \mu_-) \ge 1/2$, and hence by \cref{eq:tvub} that
	\begin{align}
		Q \geq \frac{1}{6\eta^2} = \frac{1}{1536}\frac{\lambda_{K_*}^2}{\epsilon}.
	\end{align}
	It remains to show \cref{eq:numbound}. We first apply the triangle inequality
	\begin{align}\label{eq:close1}
		\EE_{Y\sim\mu_\pm} \big[ \abs{Y - e^{-i\phi_\pm}}\big] &\leq \EE_{s\sim \cP_\pm, r|s} \big[ \big| \bra{1_{\rm tail}}\cE_r\lr{\ketbra{1_{\rm tail}}{0^L}}\ket{0^L} - \mathrm{exp}\Big(-i\sum_{j=K_*+1}^L s_j a_j\Big) \big|\big]\notag\\
		&\qquad + \EE_{s \sim \cP_\pm} \big[\big| \mathrm{exp}\Big(-i\sum_{j=K_*+1}^L s_j a_j\Big) - \exp(-i\phi_\pm) \big|\big].
	\end{align}
	The second term is easy to bound: using $|e^{-ix} - e^{-iy}| \leq |x-y|$, Jensen's inequality gives
	\begin{align}\label{eq:phasep}
		\EE_{s\sim \cP_\pm} \,\big[\big| \mathrm{exp}\Big(-i\sum_{j=K_*+1}^L s_j a_j\Big) - \exp(-i\phi_\pm) \big|\big] &\leq \EE_{s\sim \cP_\pm}\,\big[\big| \sum_{j=K_*+1}^L s_j a_j - \phi_\pm \big|\big] \leq \Big( \mathrm{Var}_{\cP_\pm} \Big[ \sum_{j=K_*+1}^L s_j a_j \Big]\Big)^{1/2},
	\end{align}
	where the variance is bounded by
	\begin{align}\label{eq:varp}
		\mathrm{Var}_{\cP_\pm} \Big[ \sum_{j=K_*+1}^L s_j a_j \Big] = \lr{1 - \eta^2} \sum_{j=K_*+1}^L a_j^2 \leq \sum_{j=K_*+1}^L a_j^2 \leq \epsilon
	\end{align}
	by \Cref{lemma: optimal_k}. Combining \cref{eq:varp} with \cref{eq:close1} gives
	\begin{align}\label{eq:close2}
		\EE_{Y\sim\mu_\pm} \big[ \big| Y - e^{-i\phi_\pm} \big| \big] \leq \EE_{s\sim \cP_\pm, r|s} \big[\big| \bra{1_{\rm tail}}\cE_r\lr{\ketbra{1_{\rm tail}}{0^L}}\ket{0^L} - \mathrm{exp}\Big(-i\sum_{j=K_*+1}^L s_j a_j\Big) \big|\big] + \sqrt\epsilon.
	\end{align}
	We now bound the remaining expectation.
    Similarly to the gate bound, we write the off-diagonal matrix element $\ketbra{1_{\rm tail}}{0^L}$ as $\frac{1}{2} \sum_{c=0}^3 i^{-c} \ketbra{\psi_c}$ for $\ket{\psi_c} = (| 0^L \rangle + i^{c} | 1_{\rm tail} \rangle)/\sqrt{2}$ to conclude by \cref{eq:idealphase} that
	\begin{align}\label{eq:tlb}
		\Big| \EE_{r\mid s}\big[\bra{1_{\rm tail}}\cE_r\lr{\ketbra{1_{\rm tail}}{0^L}}\ket{0^L} - \mathrm{exp}\Big(-i\sum_{j=K_*+1}^L s_j a_j \Big) \big] \Big| &\leq \norm{\cE_s\lr{\ketbra{1_{\rm tail}}{0^L}} - \cU_s\lr{\ketbra{1_{\rm tail}}{0^L}}}_1\notag \\
        &\leq \frac{1}{2}\sum_{c=0}^3 \norm{\cU_s(\ketbra{\psi_c}) - \cE_s(\ketbra{\psi_c})}_1,
	\end{align}
    where we used $\cE_s=\EE_{r\mid s}[\cE_r]$.
	Since $|\bra{1_{\rm tail}}\cE_r\lr{\ketbra{1_{\rm tail}}{0^L}}\ket{0^L}|\leq 1$ and $|\mathrm{exp}\big(-i\sum_{j=K_*+1}^L s_j a_j\big)|=1$, the same second-moment argument as \cref{eq:adfirst}--\cref{eq:adlast} gives
	\begin{align}
		\EE_{r\mid s}\big[\big| \bra{1_{\rm tail}}\cE_r\lr{\ketbra{1_{\rm tail}}{0^L}}\ket{0^L}-\mathrm{exp}\Big(-i\sum_{j=K_*+1}^L s_j a_j\Big)\big|^2\big] &\leq 2\big|\mathrm{exp}\Big(-i\sum_{j=K_*+1}^L s_j a_j\Big)-\EE_{r\mid s}\bra{1_{\rm tail}}\cE_r\lr{\ketbra{1_{\rm tail}}{0^L}}\ket{0^L}\big| \notag \\
        &\leq \sum_{c=0}^3 \norm{\cU_s(\ketbra{\psi_c}) - \cE_s(\ketbra{\psi_c})}_1.
	\end{align}
    Since $\cP=(\cP_++\cP_-)/2$, we have by the assumption $\frac{1}{2}\EE_{s\sim\cP}\norm{\cU_s(\rho) - \cE_s(\rho)}_1 \leq \epsilon$ that
    \begin{align}
        \sum_{\nu\in\{\pm\}} \EE_{s\sim\cP_\nu,r\mid s}\big[\big| \bra{1_{\rm tail}}\cE_r\lr{\ketbra{1_{\rm tail}}{0^L}}\ket{0^L}&-\mathrm{exp}\Big(-i\sum_{j=K_*+1}^L s_j a_j\Big)\big|^2\big] \notag\\
        &\leq \sum_{\nu\in\{\pm\}} \sum_{c=0}^3 \EE_{s\sim\cP_\nu} \norm{\cU_s(\ketbra{\psi_c}) - \cE_s(\ketbra{\psi_c})}_1 \leq 16\epsilon.
    \end{align}
	Plugging this into \cref{eq:close2} with Cauchy--Schwarz gives
	\begin{align}
		&\sum_{\nu\in\{\pm\}}\EE_{Y\sim\mu_\nu} \big[\big| Y - e^{-i\phi_\nu}\big|\big] \leq \big(2\sum_{\nu\in\{\pm\}} \EE_{s\sim\cP_\nu,r\mid s}\big[\big| \bra{1_{\rm tail}}\cE_r\lr{\ketbra{1_{\rm tail}}{0^L}}\ket{0^L}-\mathrm{exp}\Big(-i\sum_{j=K_*+1}^L s_j a_j\Big)\big|^2\big]\big)^{1/2} + 2\sqrt \epsilon \notag\\
        &\leq (\sqrt{32}+2)\sqrt{\epsilon} < 8\sqrt{\epsilon}.
	\end{align}
	This proves \cref{eq:numbound} and completes the proof with $c_t = 1/1536$.
\end{proof}

Next, we prove the head lower bound.

\begin{lemma}[Query complexity head lower bound]\label{lemma: query_head_lb}
Let $\epsilon \in (0, \epsilon_0)$ for sufficiently small constant $\epsilon_0 > 0$. Let $\cU_s$ denote time evolution for unit time under $H_s = \sum_{j=1}^L s_j a_j \ketbra{1}_j$. Define $\cP$ and $K_*$ as in \Cref{lemma: query_tail_lb}. Suppose an algorithm with classical oracle access to $O_s$ outputs a channel $\cE_s$ satisfying $\frac{1}{2}\EE_{s\sim\cP}\norm{\cU_s(\rho) - \cE_s(\rho)}_1 \leq \epsilon$ for all density matrices $\rho$. Then the algorithm must make $Q$ queries to $O_s$ for $Q$ satisfying
\begin{align}
    Q + c_{h,1} \ge c_{h,2} K_* - c_{h,3} \frac{\lambda_{K_*}^2}{\epsilon}
\end{align}
for constants $c_{h,1}, c_{h,2}, c_{h,3} > 0$.
\end{lemma}
\begin{proof}
If $K_*=0$, the claim is trivial. Otherwise, draw $s\sim\cP$.
Under $\cP$, the signs $s_1,\ldots,s_{K_*}$ are independent and
uniform, and they are independent of the tail signs and of the hidden
choice between $\cP_+$ and $\cP_-$. Ideal time evolution satisfies
\begin{align}
	\ket{\psi_s} = e^{-iH_s} | +^{K_*}0^{L-K_*} \rangle = \bigotimes_{j=1}^{K_*} \frac{\ket{0}+e^{-is_ja_j}\ket{1}}{\sqrt 2} \otimes |0^{L-K_*}\rangle.
\end{align}
Using $\cE_s = \Esub{r\mid s}{\cE_r}$ for each fixed $s$, our assumption $\frac{1}{2}\EE_{s\sim\cP}\norm{\cU_s(\rho) - \cE_s(\rho)}_1 \leq \epsilon$ implies 
\begin{align}
	\EE_{s,r\mid s}[ \bra{\psi_s} \cE_r\lr{\ketbra{+^{K_*}0^{L-K_*}}}\ket{\psi_s}] \geq 1 - \epsilon,
\end{align}
as can be seen by measuring the POVM $\{\ketbra{\psi_s}, I-\ketbra{\psi_s}\}$ for each fixed $s$ and then averaging over $r\mid s$. Conditioned on a transcript $r$, the remaining head signs are uniformly random; the expectation over these signs satisfies
\begin{align}\label{eq:qeps}
	\EE_{s\mid r} [\bra{\psi_s} \cE_r\lr{\ketbra{+^{K_*}0^{L-K_*}}}\ket{\psi_s} ] \leq \| \EE_{s\mid r}[\ketbra{\psi_s}] \|,
\end{align}
by H\"{o}lder's inequality. For a single unqueried index $j$, the expected state is
\begin{align}
	\frac{1}{2}\sum_{s_j = \pm 1} \lr{\frac{\ket{0} + e^{-is_ja_j}\ket{1}}{\sqrt 2}} \lr{\frac{\bra{0} + e^{is_ja_j}\bra{1}}{\sqrt 2}}
\end{align}
with eigenvalues $\cos^2(a_j/2)$ and $\sin^2(a_j/2)$. Since $a_j \leq \pi/2$ by assumption, this combines with \cref{eq:qeps} to give
\begin{align}
	\epsilon \geq \mathbbm{E}_{r} \Big[ 1 - \prod_{j \in [K_*] \setminus \tau_r} \cos^2 \frac{a_j}{2}\Big] \geq \mathbbm{E}_{r} \Big[1 - \mathrm{exp}\Big(-\frac{1}{4}\sum_{j \in [K_*] \setminus \tau_r} a_j^2\Big)\Big] \geq \frac{1}{2} \mathbbm{E}_{r} \Big[\mathrm{min} \Big(1, \frac{1}{4}\sum_{j \in [K_*] \setminus \tau_r} a_j^2\Big)\Big],
\end{align}
where we first use $\cos^2(x/2) \leq e^{-x^2/4}$ for all $x \leq \pi/2$, then $1-e^{-x} \geq \min(x,1)/2$ for $x \geq 0$. Since $a_j \geq a_{K_*}$, we have
\begin{align}
	\frac{1}{2} \EE_r \Big[ \min\Big(1, \frac{a_{K_*}^2}{4} \left|[K_*] \setminus \tau_r\right|\Big) \Big] \leq \epsilon.
\end{align}
By Markov's inequality, the probability that the min is 1 is at most $2\epsilon$; consequently,
\begin{align}
	\EE_r[\left|[K_*] \setminus \tau_r\right|] \leq 2\epsilon K_* + \frac{8\epsilon}{a_{K_*}^2}
\end{align}
giving query bound
\begin{align}
	Q \geq \EE_r[\left|[K_*] \cap \tau_r\right|] \geq (1-2\epsilon)K_* - \frac{8\epsilon}{a_{K_*}^2} \geq \frac{K_*}{2} - \frac{8\epsilon}{a_{K_*}^2},
\end{align}
where in the last inequality we take $\epsilon \leq 1/4$. Finally, we use $\epsilon \leq a_{K_*}^2 + 2a_{K_*}\lambda_{K_*}$ from \Cref{lemma: optimal_k}. If $\epsilon/a_{K_*}^2\leq 2$, then this term is already bounded by a constant. Otherwise $a_{K_*}^2<\epsilon/2$, and hence $2a_{K_*}\lambda_{K_*}\geq \epsilon-a_{K_*}^2>\epsilon/2$, which implies
\begin{align}
    \frac{\epsilon}{a_{K_*}^2} \leq 16\frac{\lambda_{K_*}^2}{\epsilon}.
\end{align}
This proves
\begin{align}
	Q + 16 \geq \frac{1}{2}K_* - 128\frac{\lambda_{K_*}^2}{\epsilon},
\end{align}
i.e., the claimed bound with $c_{h,1}=16$, $c_{h,2}=1/2$, and $c_{h,3}=128$.
\end{proof}
Combining Lemmas~\ref{lemma: query_tail_lb} and~\ref{lemma: query_head_lb} (as well as the reduction Lemma~\ref{lemma:time_dep_to_ind_reduction} from time-dependent to time-independent lower bounds) proves Theorem~\ref{thm:query}.

\end{document}